\documentclass[11pt]{article}
\usepackage[T1]{fontenc}
\usepackage[margin=1in]{geometry}
\usepackage{biolinum}
\usepackage{libertineMono}
\usepackage[utf8]{inputenc}
\usepackage{amsmath,amssymb,amsthm}
\usepackage{graphicx}
\usepackage{xcolor}
\usepackage{bbm}
\usepackage{enumerate}
\usepackage{latexsym}
\usepackage{mathtools}
\usepackage{mathpazo}
\usepackage{microtype,booktabs,array}
\usepackage[colorlinks=true,linkcolor=blue,anchorcolor=blue,citecolor=red,urlcolor=magenta,pdfusetitle]{hyperref}
\usepackage[capitalize,nameinlink,noabbrev]{cleveref}
\allowdisplaybreaks[1]
\newtheorem{theorem}{Theorem}

\newtheorem{lemma}[theorem]{Lemma} 
\newtheorem{corollary}[theorem]{Corollary}

\DeclareMathOperator{\SN}{SN}
\DeclareMathOperator{\SR}{SR}
\DeclareMathOperator{\spn}{span}
\DeclareMathOperator{\rank}{rank}
\DeclareMathOperator{\im}{im}

\DeclareMathOperator{\Density}{D}

\DeclareMathOperator{\Sym}{Sym}
\newcommand{\C}{\mathbb C}
\newcommand{\N}{\mathbb{Z}_{\geq 0}}
\newcommand{\Q}{\mathbb Q}
\newcommand{\cS}{\mathcal S}
\newcommand{\cD}{\mathcal D}
\newcommand{\bra}[1]{\langle #1\rvert}
\newcommand{\ket}[1]{\lvert #1\rangle}

\newcommand{\ketbra}[1]{\ket{#1}\bra{#1}}
\newcommand{\MPPT}{M_{\mathrm{PPT}}}
\newcommand{\doi}[1]{\href{https://doi.org/#1}{\nolinkurl{doi:#1}}}
\newcommand{\arxiv}[1]{\href{https://arxiv.org/abs/#1}{\nolinkurl{arXiv:#1}}}
\newcommand{\closure}[1]{\overline{#1}^{\,\mathrm{Zar}}}
\newcolumntype{L}[1]{>{\raggedright\arraybackslash}p{#1}}

\title{PPT states of almost maximal Schmidt number}

\author{
Nathaniel Johnston\\
\small{Mount Allison University}\\
\and
Benjamin Lovitz\\
\small{Concordia University}\\
}
\date{\today}

\begin{document}

\maketitle

\begin{abstract}
    We construct PPT states on $\C^m \otimes\C^n$ that have Schmidt number asymptotically approaching the smaller local dimension. More specifically, we construct a PPT state with Schmidt number at least
    \begin{align*}
        \left\lceil \frac{m + n - \sqrt{(m - n)^2 + 4(m + n - 1)}}{2} \right\rceil.
    \end{align*}
    In the case of equal local dimensions ($m = n$), this becomes $n - \lfloor\sqrt{2n - 1}\rfloor$, far exceeding previous constructions, which achieved $n/2 + O(1)$. In the case of unequal local dimensions, our result shows that there exists a PPT state on $\C^n \otimes \C^{3n-4}$ with Schmidt number at least $n-1$.
\end{abstract}

\section{Introduction}\label{sec:introduction}

The \textit{Schmidt number} of a bipartite quantum state is a coarse (but fundamental) measure of the entanglement present in that state~\cite{TH00,Wat18}. For example, a state has Schmidt number 1 if and only if it is separable. At the other extreme, the Schmidt number cannot exceed the minimum local dimension, and this upper bound is attained by maximally entangled states (e.g., Bell states when the local dimensions both equal $2$).

Positive-partial-transpose (PPT) states are undistillable under local operations and classical communication~\cite{HHH98}, which historically motivated the view that their entanglement is comparatively weak \cite{HLLMH18}. In fact, every separable state is PPT, and when the product of the local dimensions is $6$ or less, every PPT state is separable~\cite{Peres96,HHH96}. It is therefore natural to expect that PPT states cannot have large Schmidt number---a problem that was already considered in~\cite{SBL01}, where it was conjectured that, when both local dimensions equal $3$, PPT states cannot have Schmidt number exceeding $2$.

This conjecture was proved in~\cite{YLT16}, and it spurred significant interest in the question of how large the Schmidt number of a PPT state can be (as a function of its local dimensions). For $m,n \geq 2$, let $\MPPT(m,n)$ be the maximum Schmidt number of a PPT state acting on $\C^m \otimes \C^n$. Assuming $n \leq m$ for convenience, it trivially holds that $\MPPT(m,n) \leq n$. The result of \cite{YLT16} says exactly that $\MPPT(3,3) \leq 2$, and since there exist entangled PPT states acting on $\C^3 \otimes \C^3$ \cite{Horodecki97} we actually have $\MPPT(3,3) = 2$.

Much of the recent progress on this problem has been devoted to proving lower bounds on $\MPPT(m,n)$. It was shown in~\cite{HLLMH18} that $\MPPT(n,n) \geq \lceil n/4\rceil$ when $n$ is even, which was subsequently improved in \cite{SP18,PV19} to $\MPPT(n,n) \geq n/2$ (see also~\cite{Par26}). It was shown in~\cite{Nechita18} that, for fixed $n$, shifted Gaussian matrices are almost surely PPT with Schmidt number greater than $\lfloor (n-1)/16 \rfloor$ as $m \rightarrow \infty$. The paper~\cite{KG24} constructed PPT states of Schmidt number $(n+1)/2$ for odd $n$ with $m = (n+1)(n+3)/8$, along with a $5 \times 5$ example of Schmidt number 3. The subsequent work~\cite{KG26} gives a $4 \times 5$ example of Schmidt number 3, a $7 \times 7$ example of Schmidt number 4, and a $9 \times 9$ example of Schmidt number 5, so $\MPPT(5,5) \geq \MPPT(4,5) \geq 3$, $\MPPT(7,7) \geq 4$, and $\MPPT(9,9) \geq 5$.

Notably, none of these lower bounds were much better than $\MPPT(m,n) \geq n/2$, and the problem of constructing PPT states with Schmidt number substantially larger than $n/2$ remained open. In this work, we break this barrier and show that there exist PPT states with Schmidt number asymptotically approaching $n$ (the maximum possible Schmidt number). More specifically, we show that $\MPPT(m,n) \geq n - O(\sqrt{n})$. Equivalently, via a standard duality between Schmidt number and $k$-positivity, and between PPT states and decomposable linear maps~\cite{Sto82,Par26}, we show that there exist $k$-positive indecomposable maps acting on the set of $n \times n$ complex matrices when $k = n - O(\sqrt{n})$.

\subsection{Main results}\label{sec:main_statements}

Our main contribution is to construct PPT states of asymptotically much larger Schmidt number than previous constructions (we still take $n \leq m$ for convenience). While the previous lower bounds on $\MPPT(m,n)$ presented are all at most $n/2 + O(1)$, we give an improved lower bound of $\MPPT(m,n) \geq n - O(\sqrt{n})$. We emphasize that while our construction improves the asymptotic scaling of previous constructions, it does not improve the lower bounds on $\MPPT(m,n)$ in the small-dimensional cases mentioned above.

Let $\SN(\rho)$ be the Schmidt number of $\rho$. Our main result is the following:

\begin{theorem}\label{thm:intro_main}
    Let $m,n \geq 2$ be integers. There exists a PPT state $\rho \in \Density(\C^m \otimes \C^n)$ with
    \begin{align}\label{eq:mainlower}
        \SN(\rho) & \geq \left\lceil \frac{m + n - \sqrt{(m - n)^2 + 4(m + n - 1)}}{2} \right\rceil.
    \end{align}
\end{theorem}

In fact, our proof of this theorem is completely constructive: we make its statement more precise in \Cref{thm:main_explicit}, where we demonstrate an explicit PPT state $\rho \in \Density(\C^m \otimes \C^n)$ with Schmidt number satisfying the claimed bound.

Of particular note are the specializations of \Cref{thm:intro_main} to the cases when $m = n$ and when $m \geq 3n-4$ (these corollaries both follow trivially from \Cref{thm:intro_main}, so we do not prove them explicitly):

\begin{corollary}\label{cor:intro_equal_dims}
    Let $n \geq 2$ be an integer. There exists a PPT state $\rho \in \Density(\C^n \otimes \C^n)$ with
    \[
        \SN(\rho) \geq n - \left\lfloor \sqrt{2n-1} \right\rfloor.
    \]
    In particular, $\displaystyle \lim_{n\rightarrow\infty} \frac{\MPPT(n,n)}{n} = 1$.
\end{corollary}

\begin{corollary}\label{cor:intro_large_unequal_dims}
    Let $n \geq 2$ and $m \geq 3n-4$ be integers. There exists a PPT state $\rho \in \Density(\C^m \otimes \C^n)$ with $\SN(\rho) \geq n - 1$.
\end{corollary}


\subsection{Technical overview}

Let $\Density(\C^m \otimes \C^n)$ be the set of density matrices (quantum states) acting on $\C^m \otimes \C^n$. We denote arbitrary vectors using boldface $\mathbf{v}$ and unit vectors using kets $\ket{v}$. The \textit{Schmidt rank} of a non-zero vector $\mathbf{v} \in \C^m \otimes \C^n$, denoted by $\SR(\mathbf{v})$, is the rank of $\mathbf{v}$ when viewed as an $m \times n$ matrix in the natural way \cite{Wat18}.

For a state $\rho \in \Density(\C^m \otimes \C^n)$, we define the \textit{Schmidt number} $\SN(\rho)$ to be the minimum integer $r$ for which one can write
\[
    \rho = \sum_j p_j \ket{v_j}\bra{v_j}
\]
for some probability vector $\mathbf{p}$ and unit vectors $\ket{v_j} \in \C^m \otimes \C^n$ with $\SR(\ket{v_j}) \leq r$. We use $\rho^{\Gamma} := (I \otimes T)(\rho)$ to denote the partial transpose of $\rho$. We say that $\rho$ has \textit{positive partial transpose} if $\rho^{\Gamma}$ is positive semidefinite. As above, we assume $n \leq m$ for convenience.

Before giving a high-level overview of the construction, we first note an elementary fact: For any $r \in \{0,1,\dots, n-1\}$ and subspace $\mathcal{W} \subseteq \C^m \otimes \C^n$ of dimension less than $(m-r)(n-r)$, there exists a unit vector $\ket{v} \in \C^m \otimes \C^n$ such that $\spn(\mathcal{W},\ket{v})$ is not spanned by vectors of Schmidt rank $\leq r$. To see this, note that the set $\cD_r \subset \C^m \otimes \C^n$ of vectors of Schmidt rank $\leq r$ has dimension $mn - (m-r)(n-r)$ \cite{CMW08}, so $\dim(\mathcal{W}) + \dim(\cD_r) < mn$. It follows that $\mathcal{W} + \cD_r \neq \C^m \otimes \C^n$ (see the proof of~\Cref{lem:perturb_algebraic}).

We now describe the idea of the construction. We let $\sigma \in \Density(\C^m \otimes \C^n)$ be any state of rank $m+n-1$ for which $\sigma^{\Gamma}$ is positive {definite}. Let $L$ be the image of $\sigma$. By the dimension argument above, as long as $m+n-1 < (m-r)(n-r)$, for any $t > 0$ the support of the state
\[
    \sigma_t := \frac{\sigma + t \ket{v}\bra{v}}{1+t}
\]
is not spanned by vectors of Schmidt rank at most $r$, and hence $\SN(\sigma_t) > r$. Choosing $r$ maximal subject to the inequality $m+n-1 < (m-r)(n-r)$ produces the bound stated by \Cref{eq:mainlower}. Since $\sigma^{\Gamma}$ is positive definite, it holds that $\sigma_t$ is PPT for all sufficiently small $t$.

If one could find a state $\sigma$ with positive definite partial transpose and rank even smaller than $m+n-1$, then this would give rise to PPT states $\sigma_t$ of even larger Schmidt number. However, this is not possible: It is known that any state with positive definite partial transpose must have rank at least $m+n-1$~\cite[Lemma~3 and Section~III.C]{li2017indistinguishability}, so new ideas will be needed to construct PPT states of larger Schmidt number than that of \Cref{thm:intro_main}.

For the construction, we let $\sigma$ be the normalized projection onto the following subspace $L$. Let $\Sym^{m-1}(\C^2)\subseteq (\C^2)^{\otimes m-1}$ be the \textit{symmetric subspace} of tensors $\mathbf{v}$ for which $\mathbf{v}_{i_1,\dots,i_{m-1}}=\mathbf{v}_{i_{\pi(1)},\dots,i_{\pi(m-1)}}$ for all indices $(i_1,\dots,i_{m-1})$ and all permutations $\pi \in \mathfrak{S}_{m-1}$. Then $\dim(\Sym^{m-1}(\C^2))=m$, and under the identification $\C^m \cong \Sym^{m-1}(\C^2)$ and $\C^n \cong \Sym^{n-1}(\C^2)$, our subspace $L$ is given by $L:=\Sym^{m+n-2}(\C^2) \subseteq \Sym^{m-1}(\C^2) \otimes \Sym^{n-1}(\C^2)$. Clearly $L$ has dimension $m+n-1$, and by~\cite[Equation~(8)]{LSEBM25}, the normalized projection has positive definite partial transpose. For this subspace $L$, we specify an explicit choice of $\ket{v}$ and $t$ for which the state $\sigma_t$ is PPT and has the desired Schmidt number.

Interestingly, for $m=n$ the orthogonal complement $L^{\perp}$ coincides with the subspace $U_{2,m-1}$ studied in~\cite{derksen2025constructive}. It is not hard to see that $L^{\perp}$ is in fact an \textit{entangled subspace} (meaning it avoids all non-zero product vectors) of maximum dimension, but this property is not needed for our construction.

\subsection*{Acknowledgments}

The authors used ChatGPT-6 Astra to obtain a first version of this proof, as well as to write an initial draft of this manuscript. The authors then verified and simplified the proof, and heavily revised the initial draft. The authors take full responsibility for the correctness, exposition, and attribution in the final manuscript. N.J.\ acknowledges support from NSERC Discovery Grant number RGPIN-2022-04098. B.L.\ acknowledges support from NSERC Discovery Grant number RGPIN-2026-05413.

\section{Proof of the main result}\label{sec:main_proof}

We now prove \Cref{thm:intro_main}. The proof consists of two main ingredients: a density matrix that is low-rank but whose partial transpose is positive definite (\Cref{sec:unperturbed_separable}), and a perturbation of that density matrix that necessarily has large Schmidt number (\Cref{sec:perturb}).

\subsection{A low-rank state with positive-definite partial transpose}\label{sec:unperturbed_separable}

For an integer $d \geq 0$, let $V_d : \C^{d+1} \to \Sym^d(\C^2)$ be the isometry defined by
\begin{align}\label{eq:Vd_defn}
    V_d\ket{j} = \sqrt{\frac{1}{\binom{d}{j}}}\sum_{\substack{I \in \{0,1\}^d \\ |I|=j}} \ket{i_1} \otimes \ket{i_2} \otimes \cdots \otimes \ket{i_d}
\end{align}
for $j \in \{0,1,\ldots,d\}$. That is, $V_d\ket{j}$ are the normalized Dicke states, which form an orthonormal basis of $\Sym^d(\C^2)$.

Write $P_d := V_dV_d^*$ for the orthogonal projection onto $\Sym^d(\C^2) \subseteq (\C^2)^{\otimes d}$. We can now introduce the state that we will need:

\begin{lemma}\label{lem:low_rank_pt_high_rank}
    Let $m,n \geq 2$ be integers and define the state $\sigma_{m,n} \in \Density(\C^m \otimes \C^n)$ by
    \[
        \sigma_{m,n} := \frac{1}{m+n-1}\big(V_{m-1}^* \otimes V_{n-1}^*\big) P_{m+n-2} \big(V_{m-1}\otimes V_{n-1}\big),
    \]
    where $V_{m-1}$ and $V_{n-1}$ are as in \Cref{eq:Vd_defn} and $P_{m+n-2} := V_{m+n-2}V_{m+n-2}^*$. Then
    \begin{align}\label{eq:lam_min_sig}
        \rank(\sigma_{m,n}) = m + n - 1 \quad \text{and} \quad \lambda_{\min}\big(\sigma_{m,n}^\Gamma\big) = \frac{1}{(m+n-1)\binom{m+n-2}{m-1}}.
    \end{align}
\end{lemma}

\begin{proof}
    The claim that $\rank(\sigma_{m,n}) = m + n - 1$ follows from the facts that $\rank(\sigma_{m,n}) = \rank(P_{m+n-2})$ (since $\Sym^{m+n-2}(\C^2) \subseteq \Sym^{m-1}(\C^2) \otimes \Sym^{n-1}(\C^2)$) and $\rank(P_{m+n-2}) = m + n - 1$. This also shows that $\sigma_{m,n}$ is indeed a state as claimed: since $P_{m+n-2}$ is an orthogonal projection of rank $m+n-1$, $P_{m+n-2}/(m+n-1)$ is a density matrix, and $\sigma_{m,n}$ is isometrically equivalent to it.
    
    The claim about $\lambda_{\min}\big(\sigma_{m,n}^\Gamma\big)$ comes from \cite[Equation~(8)]{LSEBM25}.
\end{proof}

In particular, notice that \Cref{eq:lam_min_sig} implies that the minimal eigenvalue of $\sigma_{m,n}^\Gamma$ is strictly positive, so it has full rank, despite $\sigma_{m,n}$ itself having very low rank.

\subsection{Perturbation}\label{sec:perturb}

In this section, we introduce a direction in which we can perturb the state $\sigma_{m,n}$ from \Cref{sec:unperturbed_separable} so as to construct our desired high-Schmidt-number PPT quantum state. We let $\Q$ denote the rational numbers and let $\overline{\Q}$ be the algebraic numbers. To start, consider the $mn$ real numbers
\begin{align}\label{eq:thetaG}
    g_{i,j} = \exp\big(2^{(in+j)/(mn)}\big)
\end{align}
for $i \in \{0,1,2,\ldots,m-1\}$ and $j \in \{0,1,2,\ldots,n-1\}$.

\begin{lemma}\label{lem:alg_ind}
    The set of $mn$ numbers $\{g_{i,j}\}$ defined by \Cref{eq:thetaG} is algebraically independent over $\overline{\Q}$.
\end{lemma}

\begin{proof}
    The polynomial $x^{mn} - 2$ is irreducible over $\Q$ by Eisenstein's criterion (with the prime $2$), so it is the minimal polynomial of $\theta := 2^{1/(mn)}$. Since this polynomial has degree $mn$, the numbers $1$, $\theta$, $\theta^2$, $\ldots$, $\theta^{mn-1}$ are linearly independent over $\Q$ (since if they weren't linearly independent over $\Q$ then there would exist a degree-$(mn-1)$ polynomial with $\theta$ as a root, contradicting minimality of $x^{mn} - 2$).

    Suppose (for the sake of establishing a contradiction) that $\{g_{i,j}\}$ is not algebraically independent over $\overline{\Q}$. Then there exists a nonzero polynomial $p \in \overline{\Q}[x_0,x_1,\ldots,x_{mn-1}]$ such that
    \[
        p\big(\exp(\theta^0), \exp(\theta^1), \ldots, \exp(\theta^{mn-1})\big) = 0.
    \]
    Write $p$ as a finite sum of distinct monomials:
    \[
        p(x_0, x_1, \ldots, x_{mn-1}) = \sum_{\boldsymbol{\nu}} c_{\boldsymbol{\nu}} \prod_{q=0}^{mn-1} x_q^{\nu_q},
    \]
    where the sum is over all tuples $\boldsymbol{\nu} = (\nu_0, \nu_1, \ldots, \nu_{mn-1}) \in \N^{mn}$, $c_{\boldsymbol{\nu}} \in \overline{\Q}$ for all $\boldsymbol{\nu}$, and at least one (and only finitely many) of the coefficients $c_{\boldsymbol{\nu}}$ is nonzero. Plugging in $x_q = \exp(\theta^q)$ gives
    \begin{align}\label{eq:alg_indep_equation}
        0 & = \sum_{\boldsymbol{\nu}} c_{\boldsymbol{\nu}}\prod_{q=0}^{mn-1} \exp(\theta^q)^{\nu_q} = \sum_{\boldsymbol{\nu}} c_{\boldsymbol{\nu}}\exp\left( \sum_{q=0}^{mn-1} \nu_q\theta^q \right).
    \end{align}

    We claim that the quantities
    \[
        \sum_{q=0}^{mn-1} \nu_q\theta^q
    \]
    are all algebraic and distinct. Indeed, they are algebraic since $\theta$ is algebraic (and $\nu_q \in \N$ for all $q$). To see that they are distinct, notice that if two multi-indices $\boldsymbol{\nu}$ and $\boldsymbol{\mu}$ gave the same quantity then we would have
    \[
        \sum_{q=0}^{mn-1}(\nu_q-\mu_q)\theta^q = 0.
    \]
    However, linear independence of $1, \theta, \theta^2, \ldots, \theta^{mn-1}$ over $\Q$ then implies that $\nu_q = \mu_q$ for every $q$, so $\boldsymbol{\nu} = \boldsymbol{\mu}$.
    
    It follows that \Cref{eq:alg_indep_equation} is a nontrivial linear equation over $\overline{\Q}$ of exponentials of distinct algebraic numbers. This contradicts the Lindemann--Weierstrass theorem (see \cite{Pop24}, for example). We thus conclude that $\{g_{i,j}\}$ is algebraically independent over $\overline{\Q}$, as desired.
\end{proof}

Our main technical result of this section shows that any vector in $\C^m \otimes \C^n$ whose entries are the numbers $\{g_{i,j}\}$ (in any order) must ``point outside of'' any low-dimensional subspace with an algebraic basis:

\begin{lemma}\label{lem:perturb_algebraic}
    Let $m,n \geq 2$ be positive integers, let $\cS \subseteq \C^m \otimes \C^n$ be a subspace with a basis whose entries all belong to $\overline{\Q}$, and suppose $r \in \{0, 1, 2, \ldots, \min\{m,n\}-1\}$ satisfies
    \begin{align}\label{eq:dimension_condition}
        (m-r)(n-r) > \dim(\cS).
    \end{align}
    Suppose further that $\mathbf{v} \in \C^m \otimes \C^n$ is a vector with the numbers from \Cref{eq:thetaG} as its $mn$ entries (in any order), and that $z \in \C$ and $\mathbf{w} \in \cS$ are such that
    \begin{align}\label{eq:perturb_avoid}
        \SR(\mathbf{w} + z\mathbf{v}) \leq r.
    \end{align}
    Then $z = 0$. Furthermore, $\mathbf{v} \notin \cS$.
\end{lemma}

\begin{proof}
    We refer the reader to~\cite{Har92} for the algebraic geometry background used in this proof. Let
    \[
        \cD_r := \{\mathbf{x} \in \C^m \otimes \C^n : \SR(\mathbf{x}) \leq r\}.
    \]
    Under the usual identification of $\C^m \otimes \C^n$ with the set of $m \times n$ complex matrices, this is the determinantal variety of matrices of rank at most $r$. In particular, $\dim(\cD_r) = r(m+n-r)$, and $\cD_r$ is the common zero set of the $(r+1) \times (r+1)$ minors of the corresponding coefficient matrix (see e.g.~\cite[Proposition 12.2]{Har92} or~\cite[Lemma~4]{CMW08}).

    Consider the algebraic set
    \[
        Z := \closure{\cS + \cD_r},
    \]
    where the ``$+$'' here is the usual Minkowski sum and ``$\closure{\phantom{\cS}}$'' denotes the Zariski closure. Since the map $F : \cS \times \cD_r \rightarrow \cS + \cD_r \subseteq \C^m \otimes \C^n$ defined by $F(\mathbf{a},\mathbf{b}) = \mathbf{a} + \mathbf{b}$ is a polynomial map, we have
    \begin{align*}
        \dim(Z) & \leq \dim(\cS) + \dim(\cD_r) = \dim(\cS) + r(m+n-r) < mn,
    \end{align*}
    where the first inequality follows from e.g.~\cite[Theorem 11.12]{Har92}), and the second follows from~\Cref{eq:dimension_condition}. It follows that $Z$ is a proper algebraic subset of $\C^m \otimes \C^n$.
    
    We claim that $Z$ is cut out by polynomials with coefficients in $\overline{\Q}$. To see this, it suffices to notice that $\cS$ is cut out by polynomials with coefficients in $\overline{\Q}$ (since it has a basis with entries in this set), $\cD_r$ is cut out by polynomials with coefficients in $\Q$ (in fact, these polynomials are determinants, which have coefficients in $\{-1,0,1\}$), and $F$ has coefficients in $\Q$.
    
    Now suppose that \Cref{eq:perturb_avoid} holds for some $z \neq 0$. Then $\mathbf{w} + z\mathbf{v} \in \cD_r$, so
    \[
        \mathbf{v} = \frac{\mathbf{w}+z\mathbf{v}}{z} - \frac{\mathbf{w}}{z} \in \cD_r + \cS \subseteq Z.
    \]
    However, \Cref{lem:alg_ind} tells us that the $mn$ entries of $\mathbf{v}$ are algebraically independent over $\overline{\Q}$. As a result, $\mathbf{v}$ cannot belong to $Z$, since $\dim(Z) < mn$ (membership in $Z$ would imply the existence of a nonzero polynomial $p$ with coefficients in $\overline{\Q}$ such that $p(\mathbf{v}) = 0$). We thus conclude that $z \neq 0$ is impossible, so $z = 0$.

    Finally, since $\cS \subseteq Z$, the fact that $\mathbf{v} \notin Z$ implies $\mathbf{v} \notin \cS$, which completes the proof.
\end{proof}

\subsection{Finishing the proof}\label{sec:finalize_proof}

We now put together the ingredients from the two previous subsections in order to prove our main result. In particular, the following theorem immediately implies \Cref{thm:intro_main}:

\begin{theorem}\label{thm:main_explicit}
    Let $m,n \geq 2$ be positive integers, let $\sigma_{m,n} \in \Density(\C^m \otimes \C^n)$ be as in \Cref{lem:low_rank_pt_high_rank}, and let $\mathbf{v} \in \C^m \otimes \C^n$ be a vector with the numbers from \Cref{eq:thetaG} as its $mn$ entries (in any order). Define
    \begin{align}\label{eq:explicitg}
        C := \frac{1}{(m+n-1)\binom{m+n-2}{m-1}} \quad \text{and} \quad \ket{g} := \frac{\mathbf{v}}{\|\mathbf{v}\|}.
    \end{align}
    Then the state
    \[
        \rho := \frac{1}{2C+1}\sigma_{m,n} + \frac{2C}{2C+1}\ketbra{g}
    \]
    is PPT and has
    \begin{align}\label{eq:SN_bound}
        \SN(\rho) \geq \left\lceil \frac{m + n - \sqrt{(m - n)^2 + 4(m + n - 1)}}{2} \right\rceil.
    \end{align}
\end{theorem}

\begin{proof}
    We first show that $\rho$ is PPT. For every unit vector $\ket{g}$, $\lambda_{\min}\big((\ketbra{g})^\Gamma\big) \geq -1/2$. Indeed, if the Schmidt coefficients of $\ket{g}$ are $\gamma_1 \geq \gamma_2 \geq \cdots \geq \gamma_{\min\{m,n\}} \geq 0$ then the negative eigenvalues of $(\ketbra{g})^\Gamma$ are $-\gamma_i\gamma_j$ \cite[Lemma~1]{JP18}, which are no smaller than $-1/2$ (since $2\gamma_i\gamma_j \leq \gamma_i^2 + \gamma_j^2 \leq 1$). Together with $\lambda_{\min}\big(\sigma_{m,n}^\Gamma) = C$ from \Cref{eq:lam_min_sig}, this gives
    \begin{align*}
        \lambda_{\min}\big(\rho^\Gamma\big) & \geq \frac{1}{2C+1}\lambda_{\min}\big(\sigma_{m,n}^\Gamma\big) + \frac{2C}{2C+1}\lambda_{\min}\big((\ketbra{g})^\Gamma\big) \geq \frac{C}{2C+1} - \frac{C}{2C+1} = 0,
    \end{align*}
    so $\rho$ has positive partial transpose.

    To bound the Schmidt number of $\rho$, first notice that if
    \begin{align}\label{eq:r_value}
        r := \left\lceil \frac{m + n - \sqrt{(m - n)^2 + 4(m + n - 1)}}{2} \right\rceil - 1
    \end{align}
    then $(m-r)(n-r) > m + n - 1 = \rank(\sigma_{m,n})$ (indeed, solving the quadratic equation $(m-r)(n-r) = m + n - 1$ for $r$ gives exactly the quantity inside the ceiling in \Cref{eq:r_value}, and our choice of $r$ is strictly smaller than this quantity).
    
    Now suppose (for the sake of establishing a contradiction) that $\SN(\rho) \leq r$, so that
    \begin{align}\label{eq:proof_rho_sr_r_decomp}
        \rho = \sum_j p_j \ketbra{x_j}
    \end{align}
    for some probability vector $\mathbf{p}$ and unit vectors $\ket{x_j} \in \im(\rho) = \spn(\im(\sigma_{m,n}), \ket{g})$ with $\SR(\ket{x_j}) \leq r$ for all $j$. Since the entries of $\sigma_{m,n}$ are all in $\overline{\Q}$, $\im(\sigma_{m,n})$ has a basis consisting of vectors with entries from $\overline{\Q}$. \Cref{lem:perturb_algebraic} (applied to $\cS = \im(\sigma_{m,n})$) then shows that $\ket{x_j} \in \im(\sigma_{m,n})$ for all $j$. However, the decomposition of \Cref{eq:proof_rho_sr_r_decomp} then implies that $\im(\rho) \subseteq \im(\sigma_{m,n})$, which contradicts the fact that $\ket{g} \notin \im(\sigma_{m,n})$ (which also comes from \Cref{lem:perturb_algebraic}). This completes the proof.
\end{proof}

The particular scalars that were used in the definition of $\rho$ in \Cref{thm:main_explicit} are only required to show that $\rho$ is PPT; the same Schmidt number guarantee holds for every nontrivial convex combination of $\sigma_{m,n}$ and $\ketbra{g}$. As a result, the same proof shows that the state
\[
    \rho_p := (1-p)\sigma_{m,n} + p\ketbra{g}
\]
is PPT and satisfies the Schmidt number bound of \Cref{eq:SN_bound} for all $p \in (0,2C/(2C+1)]$.

\bibliographystyle{alpha}
\bibliography{references}

@article{TH00,
  author = {Terhal, Barbara M. and Horodecki, Pawe{\l}},
  title = {{Schmidt} number for density matrices},
  journal = {Physical Review A}, volume = {61}, pages = {040301(R)}, year = {2000},
  doi = {10.1103/PhysRevA.61.040301}, eprint = {quant-ph/9911117}, archivePrefix = {arXiv},
  note = {\doi{10.1103/PhysRevA.61.040301}; \arxiv{quant-ph/9911117}}
}

@article{Peres96,
  author = {Peres, Asher},
  title = {Separability criterion for density matrices},
  journal = {Physical Review Letters}, volume = {77}, pages = {1413--1415}, year = {1996},
  doi = {10.1103/PhysRevLett.77.1413},
  note = {\doi{10.1103/PhysRevLett.77.1413}; \arxiv{quant-ph/9604005}}
}

@article{HHH96,
  author = {Horodecki, Micha{\l} and Horodecki, Pawe{\l} and Horodecki, Ryszard},
  title = {Separability of mixed states: Necessary and sufficient conditions},
  journal = {Physics Letters A}, volume = {223}, pages = {1--8}, year = {1996},
  doi = {10.1016/S0375-9601(96)00706-2},
  note = {\doi{10.1016/S0375-9601(96)00706-2}; \arxiv{quant-ph/9605038}}
}

@article{Horodecki97,
  author = {Horodecki, Pawe{\l}},
  title = {Separability criterion and inseparable mixed states with positive partial transposition},
  journal = {Physics Letters A},
  volume = {232},
  pages = {333--339},
  year = {1997},
  doi = {10.1016/S0375-9601(97)00416-7},
  note = {\doi{10.1016/S0375-9601(97)00416-7}; \arxiv{quant-ph/9703004}}
}

@article{HHH98,
  author = {Horodecki, Micha{\l} and Horodecki, Pawe{\l} and Horodecki, Ryszard},
  title = {Mixed-state entanglement and distillation: Is there a ``bound'' entanglement in nature?},
  journal = {Physical Review Letters}, volume = {80}, pages = {5239--5242}, year = {1998},
  doi = {10.1103/PhysRevLett.80.5239},
  note = {\doi{10.1103/PhysRevLett.80.5239}; \arxiv{quant-ph/9801069}}
}

@article{SBL01,
  author = {Sanpera, Anna and Bru{\ss}, Dagmar and Lewenstein, Maciej},
  title = {{Schmidt}-number witnesses and bound entanglement},
  journal = {Physical Review A},
  volume = {63},
  pages = {050301(R)},
  year = {2001},
  doi = {10.1103/PhysRevA.63.050301},
  note = {\doi{10.1103/PhysRevA.63.050301}; \arxiv{quant-ph/0009109}}
}

@article{YLT16,
  author = {Yang, Yu and Leung, Denny H. and Tang, Wai-Shing},
  title = {All {$2$}-positive linear maps from {$M_3(\mathbb C)$} to {$M_3(\mathbb C)$} are decomposable},
  journal = {Linear Algebra and its Applications}, volume = {503}, pages = {233--247}, year = {2016},
  doi = {10.1016/j.laa.2016.03.050},
  note = {\doi{10.1016/j.laa.2016.03.050}; \arxiv{1603.03534}}
}

@article{HLLMH18,
  author = {Huber, Marcus and Lami, Ludovico and Lancien, C{\'e}cilia and M{\"u}ller-Hermes, Alexander},
  title = {High-dimensional entanglement in states with positive partial transposition},
  journal = {Physical Review Letters}, volume = {121}, pages = {200503}, year = {2018},
  doi = {10.1103/PhysRevLett.121.200503},
  note = {\doi{10.1103/PhysRevLett.121.200503}; \arxiv{1802.04975}}
}

@article{SP18,
  author = {Sindici, Enrico and Piani, Marco},
  title = {Simple class of bound entangled states based on the properties of the antisymmetric subspace},
  journal = {Physical Review A}, volume = {97}, pages = {032319}, year = {2018},
  doi = {10.1103/PhysRevA.97.032319},
  note = {\doi{10.1103/PhysRevA.97.032319}; \arxiv{1708.06595}}
}

@article{PV19,
  author = {P{\'a}l, K{\'a}roly F. and V{\'e}rtesi, Tam{\'a}s},
  title = {Class of genuinely high-dimensionally entangled states with a positive partial transpose},
  journal = {Physical Review A}, volume = {100}, pages = {012310}, year = {2019},
  doi = {10.1103/PhysRevA.100.012310},
  note = {\doi{10.1103/PhysRevA.100.012310}; \arxiv{1904.08282}}
}

@article{KG24,
  author = {Krebs, Robin and Gachechiladze, Mariami},
  title = {High {Schmidt} number concentration in quantum bound entangled states},
  journal = {Physical Review Letters}, volume = {132}, pages = {220203}, year = {2024},
  doi = {10.1103/PhysRevLett.132.220203},
  note = {\doi{10.1103/PhysRevLett.132.220203}; \arxiv{2402.12966}}
}

@article{KG26,
  author = {Krebs, Robin and Gachechiladze, Mariami},
  title = {Scaling bound entanglement through local extensions},
  journal = {Physical Review A}, volume = {113}, pages = {062428}, year = {2026},
  doi = {10.1103/tcq7-q96m},
  note = {\doi{10.1103/tcq7-q96m}; \arxiv{2509.07086}}
}

@misc{Par26,
  author = {Park, Sang-Jun},
  title = {{$k$}-Positivity and high-dimensional bound entanglement under symplectic group symmetries},
  year = {2026}, eprint = {2602.09860}, archivePrefix = {arXiv}, primaryClass = {quant-ph},
  note = {\arxiv{2602.09860v3}}
}

@article{CMW08,
  author = {Cubitt, Toby S. and Montanaro, Ashley and Winter, Andreas},
  title = {On the dimension of subspaces with bounded {Schmidt} rank},
  journal = {Journal of Mathematical Physics}, volume = {49}, pages = {022107}, year = {2008},
  doi = {10.1063/1.2862998},
  note = {\doi{10.1063/1.2862998}; \arxiv{0706.0705}}
}

@incollection{Pop24,
  author = {Popescu, Sever Angel},
  title = {A Simple and Self-contained Proof for the {L}indemann--{W}eierstrass Theorem},
  booktitle = {New Frontiers in Number Theory and Applications},
  publisher = {Springer Nature Switzerland},
  address = {Cham},
  pages = {349--366},
  year = {2024},
  doi = {10.1007/978-3-031-51959-8_16},
  eprint = {2306.14352},
  archivePrefix = {arXiv},
  primaryClass = {math.NT},
  note = {\doi{10.1007/978-3-031-51959-8_16}; \arxiv{2306.14352}}
}

@article{Nechita18,
  author = {Nechita, Ion},
  title = {On the separability of unitarily invariant random quantum states: The unbalanced regime},
  journal = {Advances in Mathematical Physics}, volume = {2018}, pages = {7105074}, year = {2018},
  doi = {10.1155/2018/7105074},
  note = {\doi{10.1155/2018/7105074}; \arxiv{1802.00067}}
}

@article{LSEBM25,
  author = {Louvet, Jonathan and Serrano-Ens{\'a}stiga, Eduardo and Bastin, Thierry and Martin, John},
  title = {Nonequivalence between absolute separability and positive partial transposition in the symmetric subspace},
  journal = {Physical Review A},
  volume = {111},
  pages = {042418},
  year = {2025},
  doi = {10.1103/PhysRevA.111.042418},
  eprint = {2411.16461},
  archivePrefix = {arXiv},
  primaryClass = {quant-ph},
  note = {\doi{10.1103/PhysRevA.111.042418}; \arxiv{2411.16461}}
}

@book{Wat18,
  author = {Watrous, John},
  title = {The Theory of Quantum Information},
  publisher = {Cambridge University Press},
  year = {2018},
  doi = {10.1017/9781316848142},
  note = {\doi{10.1017/9781316848142}; \url{https://cs.uwaterloo.ca/~watrous/TQI/}}
}

@article{Sto82,
    author  = {St{\o}rmer, Erling},
    title   = {Decomposable Positive Maps on {$C^*$}-Algebras},
    journal = {Proceedings of the American Mathematical Society},
    volume  = {86},
    number  = {3},
    pages   = {402--404},
    year    = {1982},
    doi     = {10.1090/S0002-9939-1982-0671203-5},
    note = {\doi{10.1090/S0002-9939-1982-0671203-5}}
}

@article{JP18,
  author = {Johnston, Nathaniel and Patterson, Everett},
  title = {The inverse eigenvalue problem for entanglement witnesses},
  journal = {Linear Algebra and its Applications},
  volume = {550},
  pages = {1--27},
  year = {2018},
  doi = {10.1016/j.laa.2018.03.043},
  eprint = {1708.05901},
  archivePrefix = {arXiv},
  primaryClass = {quant-ph},
  note = {\doi{10.1016/j.laa.2018.03.043}; \arxiv{1708.05901}}
}

@misc{derksen2025constructive,
  author = {Derksen, Harm and Lovitz, Benjamin},
  title = {Constructive counterexamples to the additivity of minimum output {R}{\'e}nyi entropy of quantum channels for all $p>1$},
  year = {2025},
  eprint = {2510.07547},
  archivePrefix = {arXiv},
  primaryClass = {quant-ph},
  note = {\arxiv{2510.07547}}
}

@article{li2017indistinguishability,
  author = {Li, Yinan and Wang, Xin and Duan, Runyao},
  title = {Indistinguishability of bipartite states by positive-partial-transpose operations in the many-copy scenario},
  journal = {Physical Review A},
  volume = {95},
  number = {5},
  pages = {052346},
  year = {2017},
  doi = {10.1103/PhysRevA.95.052346},
  eprint = {1702.00231},
  archivePrefix = {arXiv},
  primaryClass = {quant-ph},
  note = {\doi{10.1103/PhysRevA.95.052346}; \arxiv{1702.00231}}
}

@book{Har92,
  author = {Harris, Joe},
  title = {Algebraic Geometry: A First Course},
  series = {Graduate Texts in Mathematics},
  volume = {133},
  publisher = {Springer},
  address = {New York},
  year = {1992},
  doi = {10.1007/978-1-4757-2189-8},
  note = {\doi{10.1007/978-1-4757-2189-8}}
}

\end{document}